%% file: LandsmanLindenhoviusRevised.tex
\documentclass[11pt]{article}
\usepackage{amsmath, amssymb,amscd}
\usepackage{amsthm}
\usepackage{mathptmx}       % selects Times Roman as basic font
\usepackage{helvet}         % selects Helvetica as sans-serif font
\usepackage{courier}        % selects Courier as typewriter font
\usepackage{makeidx}         % allows index generation
\usepackage{graphicx}        % standard LaTeX graphics tool
\usepackage{multicol}        % used for the two-column index
\usepackage[bottom]{footmisc}% places footnotes at page bottom
\usepackage{a4wide}
 \usepackage[bookmarks=false]{hyperref}

\input{BohrificationpreambleLL2.tex}

\usepackage{stmaryrd}

\begin{document}
\title{Symmetries  in exact Bohrification}
\author{\emph{Dedicated to the memory of Paul Busch (1955-2018)}\\ \ \\
	Klaas Landsman\footnote{N.P. Landsman, Institute for Mathematics, Astrophysics, and Particle Physics, Radboud University,
		Heyendaalseweg 135, 6525 AJ Nijmegen, The Netherlands,
		 email: landsman@math.ru.nl}\ \  and Bert Lindenhovius\footnote{A.J. Lindenhovius, Department of Computer Science, Tulane University, 6823 St Charles Ave, 70118 New Orleans, United States, email: alindenh@tulane.edu}   }
\date{}

\maketitle
\abstract{The `Bohrification'' program in the foundations of \qm\ implements Bohr's doctrine of classical concepts through an interplay between commutative and non-commutative operator algebras.
Following a brief  conceptual and mathematical review of this program, we focus on one half of it, called ``exact'' Bohrification, where 
a (typically noncommutative) unital $C^*$-algebra $A$ is studied through its commutative unital $C^*$-subalgebras $C\subset A$, organized into
 a poset $\CA$. 
 This poset turns out to be a rich invariant of $A$. To set the stage, we first give a general review of symmetries in elementary \qm\ (i.e., on \Hs) as well as in algebraic quantum theory, incorporating $\CA$ as a new kid in town. We then give a detailed proof of a deep result due to Hamhalter (2011), according to which $\CA$ determines $A$ as  a Jordan algebra (at least for a large class of $C^*$-algebras). As a corollary, we prove a new Wigner-type theorem to the effect that order isomorphisms of $\mathcal{C}(B(H))$ are 
 (anti) unitarily implemented. We also show how $\CA$ is related to the orthomodular poset $\CP(A)$ of projections in $A$. These results indicate that $\CA$ is  a serious player in $C^*$-algebras and quantum theory. }
\section{Bohrification}
\label{sec:1}
The Bohrification program is an attempt to relate the core of the Copenhagen Interpretation of \qm, viz.\ Bohr's
 \emph{doctrine of classical concepts}, to the mathematical formalism of operator algebras created by  von Neumann, as subsequently generalized into the theory of $C^*$-algebras by Gelfand \& Naimark (1943). Other elements of the  Copenhagen Interpretation, such as the rejection of the possibility to analyze what is going on during measurements,  the closely related idea of the collapse of the wave-function (in the sense of a ``second'' time-evolution in \qm\ beside the primary unitary evolution governed by the Schr\"{o}dinger equation), and the ensuing hybrid interpretation of quantum-mechanical states  as mere catalogues of the probabilities attached to possible outcomes of experiments, are irrelevant for this paper (and in fact appear outdated to us). To introduce the doctrine, we quote the opening of the
  most (perhaps the only) systematic presentation of the Copenhagen Interpretation by  one of its original authors (viz.\ Heisenberg):
\begin{quotation}
`The Copenhagen interpretation of quantum theory starts from a paradox. Any experiment in physics, whether it refers to the phenomena of daily life or to atomic events, is to be described in the terms of classical physics. The concepts of classical physics form the language by which we describe the arrangement of our experiments and state the results. We cannot and should not replace these concepts by any others. Still the application of these concepts is limited by the relations of uncertainty. We must keep in mind this limited range of applicability of the classical concepts while using them, but we cannot and should not try to improve them.'  (Heisenberg  1958, p.\ 44) 
\end{quotation}
  Despite their  agreement about the central role of classical concepts in the study of \qm, there seems to have been an
unresolved disagreement between Bohr and Heisenberg about their precise status (Camilleri, 2009), as follows:
\begin{itemize}
\item According to Bohr---haunted by his idea of Complementarity---only one classical concept (or sometimes one coherent family of classical concepts) applies to  the experimental study of some quantum object at a time. But if it applies, it does so exactly, and has the same meaning as in classical physics (since Bohr held that any other meaning would simply be undefined). In a different experimental setup, some other classical concept may apply, which in classical physics would have been compatible with the previous one, but in \qm\ is not. Early  examples of such ``complementary'' pairs, as presented e.g.\ in Bohr (1928),  are particle versus wave  and
space-time  versus ``causal" descriptions (by which Bohr means conservation laws). Later on, Bohr emphasized the complementarity of
 one "phenomenon" (i.e., an indivisible unit  of a quantum object coupled to an experimental arrangement) against another (cf.\ Held, 1994).
 \item Heisenberg, on the other hand, seems to have held a more relaxed attitude towards classical concepts, arguably inspired by his game-changing paper on the quantum-mechanical reinterpretation (\emph{Umdeutung}) of mechanical and kinematical relations (Heisenberg, 1925).
  In this paper, he performed the act of what we now call quantization, in putting the observables of classical physics (i.e.\ functions on a phase space) on a new mathematical footing (i.e., they were turned into matrices), where they also have new properties. In his second epoch-making paper (Heisenberg, 1927) introducing the uncertainty relations, he then tried to find some operational meaning of these ``reinterpreted'' observables through measurement procedures. Since quantization applies to all classical observables at once,  all classical concepts apply simultaneously, but approximately (ironically, Heisenberg (1925)  was inspired by Bohr's  Correspondence Principle, but later on Bohr insisted on precise nature of classical concepts described above). 
\end{itemize}

This ideological split between Bohr and Heisenberg is still with us, as it leads to a similar break-up of the Bohrification program into two parts.
The overall idea of Bohrification is to interpret classical concepts as commutative $C^*$-algebras, and hence the two parts in question  are mathematically distinguished by the specific relationship between a given 
 noncommutative $C^*$-algebra $A$ and the commutative $C^*$-algebras that give physical ``access'' to $A$. 
 Bohr's view on the precise nature of classical concepts comes back mathematically in  \emph{exact} Bohrification, which studies (unital) commutative $C^*$-subalgebras $C$ of a given (unital) noncommutative $C^*$-algebra $A$. 
 Heisenberg's interpretation of the doctrine of classical concepts, on the other hand, resurfaces in  \emph{asymptotic} Bohrification, which involves asymptotic inclusions (i.e.\ deformations) of commutative $C^*$-algebras into noncommutative ones. 
 The precise relationship between Bohr's and Heisenberg's views, and hence also between exact and asymptotic Bohrification, remains to be clarified; their joint existence is unproblematic, however, since the two programs complement each other. 
 
 As reviewed in Landsman (2016) and explained in detail in Landsman (2017), asymptotic Bohrification provides  a mathematical setting for
the measurement problem, spontaneous symmetry breaking, the classical limit of \qm, the thermodynamic limit of quantum statistical mechanics, and 
the Born rule for probabilities construed as long-run frequencies, whereas exact Bohrification  turns out to be an appropriate framework for
  Gleason's Theorem, the Kadison--Singer conjecture, the Born rule (for single case probabilities), and, initially via the  topos-theoretic approach to \qm, 
intuitionistic quantum logic. In the context of the present paper it should be mentioned that the poset $\CA$ we will be concerned with has its origins in 
 the reinterpretation of the  Kochen--Specker Theorem in the language of topos theory by Isham \& Butterfield (1998). In the setting of \vna s this led
  Hamilton, Isham, \& Butterfield (2000) to a poset similar to $\CA$ (though crucially with the opposite ordering), which was studied in great detail by
   D\"{o}ring \& Isham {(2008a-d)}.  The poset $\CA$ as we use it was introduced by Heunen, Landsman \& Spitters (2009), again in the context of topos theory.
   
In this paper we discuss the virtues of exact Bohrification in providing a new invariant $\CA$ for unital $C^*$-algebras $A$, defined as the poset of  all unital  commutative $C^*$-subalgebras of a unital $C^*$-algebra $A$ (that share the unit of $A$), ordered by inclusion. We start with a general discussion of symmetries in elementary \qm\ on \Hs\ in \S\ref{secBH}, which culminates in our Wigner Theorem for $\mathcal{C}(B(H))$. Moving to general (unital) $C^*$-algebras $A$ in \S\ref{secA}, we discuss the place of $\CA$ amidst some comparable constructions $A$ gives rise to, viz.\ its (pure) state space, its Jordan algebra structure, its effect algebra, and its (orthocomplemented) poset of projections.
In \S\ref{secHam} we give a complete and independent proof of Hamhalter's (2011) great theorem to the effect that for a large class of $C^*$-algebras $A$, order isomorphisms of 
$\CA$ are induced by Jordan automorphisms of $A$ (this theorem was predated by an analogous result by  D\"{o}ring \& Harding (2010) for \vna s, which may have been the first of its kind).
Hamhalter's  theorem is also the key lemma in our Wigner Theorem for $\mathcal{C}(B(H))$.  We close our paper  in \S\ref{secPA} with  a study of the relationship between $\CA$ and the poset $\CP(A)$ of projections in $A$.
\section{Symmetries in quantum theory on \Hs}\L{secBH}
 Even in elementary \qm, where $A=B(H)$, i.e., the $C^*$-algebra of all bounded operators on some \Hs\ $H$, the concept of a symmetry is already diverse, as least apparently, since a non-commutative $C^*$-algebra like $B(H)$ gives rise to numerous ``quantum structures''. The main examples are:
  \begin{enumerate}
\item The \hi{normal pure state space} $\CP_1(H)$, i.e., the set of one-dimensional projections on $H$, 
 with a  ``transition probability'' $\ta:\CP_1(H)\x \CP_1(H)\raw[0,1]$ defined by
\begin{equation}
\ta(e,f)=\Tr(ef). \L{taborn}
\end{equation}
\item The \hi{normal  state space} $\mathcal{D}(H)$, which is the convex set of all density operators $\rh$ on $H$
(i.e., $\rh\geq 0$ and $\Tr(\rh)=1$).
\item The \hi{self-adjoint operators} $B(H)_{\mathrm{sa}}$ on $H$, seen as a Jordan algebra.
\item The \hi{effects} $\CE(H)=[0,1]_{B(H)}$ on $H$, i.e., the set of all $a\in B(H)_{\mathrm{sa}}$ for which $0\leq a\leq 1_H$,
seen as a convex poset.
\item The \hi{projections} $\CP(H)$  on $H$, seen as an orthocomplemented lattice.
\item The \hi{unital commutative $C^*$-subalgebras $\CC(B(H))$ of $B(H)$}, seen as a poset.
\end{enumerate}
 Each structure comes with its own notion of a symmetry (whose name has been chosen for historical reasons and---except  for the first and third---is not  standard):
 \begin{definition}\L{defQMsym}
Let $H$ be a \Hs\ (not necessarily finite-dimensional). 
\begin{enumerate}
\item  A\hi{Wigner symmetry} is
a bijection  $\mathsf{W}:\CP_1(H)\raw \CP_1(H)$ that satisfies
\begin{equation}
\Tr(\mathsf{W}(e)\mathsf{W}(f))=\Tr(ef), \:\: e,f\in \CP_1(H).\L{UtrUtr}
\end{equation}
\item 
 A \hi{Kadison symmetry} is an affine bijection  $\mathsf{K}: \mathcal{D}(H)\raw \mathcal{D}(H)$.
 \item A  \hi{Jordan symmetry} is an invertible \hi{Jordan map} $\mathsf{J}: B(H)_{\mathrm{sa}}\raw B(H)_{\mathrm{sa}}$, where the latter is an
$\R$-linear map
 that satisfies either one  of the equivalent conditions
 \begin{eqnarray}
\mathsf{J}(a\circ b)&=&\mathsf{J}(a)\circ\mathsf{J}(b);\L{JPP}\\
 \mathsf{J}(a^2)&=&\mathsf{J}(a)^2,\L{Jsquares} 
\end{eqnarray}
where the \hi{Jordan product} $\circ$ is defined by $a\circ b=\half(ab+ba)$, so that $a^2=a\circ a$.
Equivalently, a Jordan symmetry is a Jordan automorphism of $B(H)$, see below. 
  \item A \hi{Ludwig symmetry} is an affine order  isomorphism $\mathsf{L}: \CE(H)\raw\CE(H)$.
 \item A \hi{von Neumann symmetry} is an order isomorphism  $\mathsf{N}: \CP(H)\raw\CP(H)$ that  preserves the orthocomplementation, i.e.\
 $\mathsf{N}(1_H-e)=1_H-\mathsf{N}(e)$, $e\in\CP(H)$.
\item A \hi{Bohr symmetry} is an order isomorphism $\mathsf{B}: \CC(B(H))\raw \CC(B(H))$.
\end{enumerate}\end{definition} 
In no.\ 2 (and 4) being \hi{affine} means that  $\mathsf{K}$ (and similarly $\mathsf{L}$) preserves convex sums, i.e., for $t\in(0,1)$ and $\rh_1,\rh_2\in \mathcal{D}(H)$,
$\mathsf{K}(t\rh_1+(1-t)\rh_2)=t\mathsf{K}\rh_1+(1-t) \mathsf{K}\rh_2$. 
In nos.\  4--6, an  \hi{order isomorphism} $\mathsf{O}$ of the given poset is a bijection such that $x\leq y$ if and only if $\mathsf{O}(x)\leq \mathsf{O}(y)$. 
In no.\ 3
one may complexify  $\mathsf{J}$ to a $\C$-linear map 
\beq
\mathsf{J}_{\C}:B(H)\raw B(H) \L{JC}
\eeq by writing $a\in B(H)$ as 
$a=b+ic$, with $b=\half (a+a^*)$ and $c=-\half i(a-a^*)$, so that $b^*=b$ and $c^*=c$, and putting
\beq
\mathsf{J}_{\C}(a)=\mathsf{J}(b)+i\mathsf{J}(c).
\eeq If $\mathsf{J}$ satisfies 
 \er{JPP} - \er{Jsquares} for each $a,b\in B(H)_{\mathrm{sa}}$, then $\mathsf{J}_{\C}$ satisfies \er{JPP} - \er{Jsquares} for each $a,b\in B(H)$ (with $\mathsf{J}\leadsto \mathsf{J}_{\C}$) as well as 
 \beq
 \mathsf{J}_{\C}(a^*)=\mathsf{J}_{\C}(a)^*. \L{JCstar}
 \eeq 
 Conversely, one may restrict such a  $\mathsf{J}_{\C}$ to the self-adjoint part 
$B(H)_{\mathrm{sa}}$ of $B(H)$, so that Jordan symmetries are essentially the same thing as Jordan automorphisms, i.e.,
$\C$-linear maps \er{JC} that satisfy \er{JCstar} and  \er{JPP} - \er{Jsquares} with $\mathsf{J}\leadsto \mathsf{J}_{\C}$.

It is well known that the first four notions of symmetry are equivalent (see for example Alfsen \& Shultz, 2001;  Bratteli \& Robinson, 1987; Cassinelli et al, 2004; Moretti, 2013).
 If $\dim(H>2)$, as a corollary to Gleason's Theorem the fifth notion is also equivalent to all of these  (Hamhalter, 2004), and, under the same assumption, so is the sixth (Hamhalter, 2011). We now sketch these equivalences; combine the above references or see Landsman (2017) for complete proofs. 
 \begin{enumerate}
\item There is a bijective correspondence between:
\begin{itemize}
\item  Wigner symmetries $\mathsf{W}:\CP_1(H)\raw \CP_1(H)$;
\item  Kadison symmetries $\mathsf{K}: \mathcal{D}(H)\raw \mathcal{D}(H)$, viz.
\end{itemize}
\begin{eqnarray}
\mathsf{W}&=&\mathsf{K}_{| \CP_1(H)};\L{UsfU1}\\
\mathsf{K}\left( \sum_i \lm_i e_{\ups_i}\right) &=& \sum_i \lm_i \mathsf{W}( \ups_{\ups_i}),\L{UsfU2}
\end{eqnarray}
where $\rh= \sum_i \lm_i e_{\ups_i}$ is a  spectral expansion of $\rh\in \mathcal{D}(H)$ in terms of a basis of eigenvector $\ups_i$ of $\rh$ with eigenvalues $\lm_i$, where
$\lm_i\geq 0$ and $\sum_i \lm_i=1$.

 It is a nontrivial fact that \er{UsfU2} is well defined (despite  non-uniqueness of the spectral expansion in case that $\rh$ has degenerate spectrum). Furthermore, 
 $\mathsf{K}$ maps $\CP_1(H)\subset  \mathcal{D}(H)$ into itself because $\CP_1(H)=\partial_e  \mathcal{D}(H)$ and affine bijections of convex sets restrict to bijections of their extreme boundaries.  Finally,
  \er{UsfU1} preserves transition probabilities because an affine bijection $\mathsf{K}: \mathcal{D}(H)\raw \mathcal{D}(H)$ extends to an isometric isomorphism  $\mathsf{K}_1: B_1(H)_{\mathrm{sa}}\raw B_1(H)_{\mathrm{sa}}$ with respect to the trace-norm $\|\cdot\|_1$, and for any $e,f\in\CP_1(H)$ we have 
  \begin{equation}
\|e-f\|_1=2\sqrt{1-\Tr(ef)}.\L{epsephvsqrt}
\end{equation}
   \item There is a bijective correspondence between:
\begin{itemize}
\item   Kadison symmetries $\mathsf{K}: \mathcal{D}(H)\raw \mathcal{D}(H)$;
\item  Jordan symmetries $\mathsf{J}: B(H)_{\mathrm{sa}}\raw B(H)_{\mathrm{sa}}$,
\end{itemize}
such that for any $a\in B(H)_{\mathrm{sa}}$ one has
\begin{equation}
\Tr (\mathsf{K}(\rh)a)=\Tr(\rh\mathsf{J}(a)). \L{Urhoa}
\end{equation}
To see this, we identify $\mathcal{D}(H)$ with the set $S_n(B(H))$ of normal states $\om$ on $B(H)$ through $\om(a)=\Tr(\rh a)$, so that with slight abuse of notation eq.\ 
\er{Urhoa}  reads 
\beq
(\mathsf{K}\om)(a)=\om(\mathsf{J}(a)). \L{simply}
\eeq This defines $\mathsf{K}$ in terms of $\mathsf{J}$. 
Conversely, we identify $B(H)_{\mathrm{sa}}$
with the set $A_b(S_n(B(H)))$ of all  real-valued bounded affine functions on the convex set $S_n(B(H))$ through the Gelfand-like transform
$a\lraw \hat{a}$, where $\hat{a}(\om)=\om(a)$; here the nontrivial analytic facts are  that the functions $\hat{a}$ exhaust $A_b(S_n(B(H)))$ and that 
$\|a\|=\|\hat{a}\|_{\infty}$. We now define  a map $$\mathsf{\hat{J}}: A_b(S_n(B(H)))\raw A_b(S_n(B(H)))$$ in terms of $\mathsf{K}$ in the obvious way, i.e., by 
$(\mathsf{\hat{J}}\hat{a})(\om)=\hat{a}(\mathsf{K}(\om))$. This, in turn, defines $\mathsf{J}$  in terms of $\mathsf{K}$, which again yields \er{simply}. 
The map  $\mathsf{\hat{J}}$ trivially preserves the (pointwise) order as well as the unit (function) in  $A_b(S_n(B(H)))$, so that the corresponding map $\mathsf{J}$
preserves the usual partial order on $\leq$ $B(H)_{\mathrm{sa}}$ (i.e.\ $a\leq b$ iff $b-a=c^2$ for some $c\in B(H)_{\mathrm{sa}}$)
as well as the unit (operator) $1_H$ in $B(H)_{\mathrm{sa}}$. Finally, for invertible linear maps these properties  are equivalent to the fact that
$\mathsf{J}$ is a Jordan symmetry. 
\item There is a bijective correspondence between:
\begin{itemize}
\item  Jordan symmetries $\mathsf{J}: B(H)_{\mathrm{sa}}\raw B(H)_{\mathrm{sa}}$;
\item   Ludwig symmetries $\mathsf{L}: \CE(H)\raw\CE(H)$.
\end{itemize}
Since $\CE(H)\subset B(H)_{\mathrm{sa}}$, we may simply restrict $\mathsf{J}$ to $\CE(H)$ so as to obtain $\mathsf{L}$.   Since a Jordan automorphism preserves order as well as the unit, the inequality $0\leq a\leq 1_H$ characterizing
 $a\in\CE(H)$ is preserved, i.e., $0\leq\mathsf{J}(a)\leq 1_H$. In other words, $\mathsf{J}$
 preserves $\CE(H)$, whose order it preserves, too. Convexity is obvious, since $\mathsf{L}=\mathsf{J}_{|\CE(H)}$ comes from a linear map. 
 Conversely, since $\mathsf{L}$ is an order isomorphism, it must satisfy $\mathsf{L}(0)=0$ (as well as $\mathsf{L}(1_H)=1_H$), since $0$ is  the bottom element of $\CE(H)$ as an ordered set (and 
 $1_H$ is its the top element). One can show that this property plus convexity yields a linear extension $\mathsf{J}$ of $\mathsf{L}$ from $\CE(H)$ to $B(H)_{\mathrm{sa}}$, which is unital as well order-preserving, and hence is a Jordan symmetry.
\item  If $\dim(H)>2$, then there is a bijective correspondence between:
\begin{itemize}
\item  Jordan symmetries $\mathsf{J}: B(H)_{\mathrm{sa}}\raw B(H)_{\mathrm{sa}}$;
\item von Neumann symmetries   $\mathsf{N}: \CP(H)\raw\CP(H)$.
\end{itemize}
Jordan symmetries restrict to order isomorphisms of $\CP(H)\subset B(H)_{\mathrm{sa}}$; the only nontrivial point is that the order in $\CP(H)$ (i.e., $e\leq f$ iff $ef=e$, which is the case iff $eH\subseteq fH$) coincides with the order inherited from $B(H)_{\mathrm{sa}}$. 
Conversely, one may attempt to extend some map $\mathsf{N}: \CP(H)\raw\CP(H)$ to $B(H)_{\mathrm{sa}}$ by first supposing that $a\in B(H)_{\mathrm{sa}}$
has a finite spectral decomposition 
 $a=\sum_j \lm_j f_j$, where  $(f_j)$  is a family of mutually orthogonal projections and  $\lm_j\in\R$, and putting
\begin{equation}
\mathsf{J}(a)=\sum_j \lm_j \mathsf{N}(f_j). \L{alfromN}
\end{equation}
For general $a$, one then hopes to be able to use the spectral theorem in order to extend $\mathsf{J}$ to all of $B(H)_{\mathrm{sa}}$ by continuity.
  It is far from trivial that this construction works and yields an $\R$-linear map, but it does. The proof relies on Gleason's Theorem (whence the assumption $\dim(H)>2$), which in turn can be invoked because  von Neumann symmetries preserve all suprema in $\CP(H)$. 
The extension $\mathsf{J}$  thus obtained is positive and unital, and hence is a Jordan symmetry.
\item  If $\dim(H)>2$, then there is a bijective correspondence between:
\begin{itemize}
\item  Jordan symmetries $\mathsf{J}: B(H)_{\mathrm{sa}}\raw B(H)_{\mathrm{sa}}$;
\item Bohr symmetries  $\mathsf{B}: \CC(B(H))\raw \CC(B(H))$.
\end{itemize}
Given $\mathsf{J}$, as explained above we first complexify it so as to obtain a Jordan automorphism $\mathsf{J}_{\C}:B(H)\raw B(H)$. It is a standard result that such maps are isometric. 
If $C\subset B(H)$ is commutative, then so is its image $\mathsf{J}_{\C}(C)$, since commutativity of $C$ is equivalent to associativity of the Jordan product within $C$, and hence is preserved under Jordan maps. 
Furthermore, since $\mathsf{J}_{\C}$ is an isometry on $C$, its image is (norm) closed, and  by \er{JCstar} it is also self-adjoint. Finally, Jordan automorphisms preserve the unit $1_H$, so that if $C$ is a unital commutative $C^*$-subalgebra of $B(H)$, then  so is $\mathsf{J}_{\C}(C)$.
Thus  $\mathsf{J}$ induces a map $\mathsf{B}$ by $\mathsf{B}(C)=\mathsf{J}_{\C}(C)$. 

  Trivially, if $C\subseteq D$ in $B(H)$, so that $C\leq D$ in $\mathcal{C}(B(H))$, then $\mathsf{J}_{\C}(C)\subseteq\mathsf{J}_{\C}(D)$ in $B(H)$, so that $\mathsf{J}(C)\leq\mathsf{J}(D)$ in $\mathcal{C}(B)$. It follows that $\mathsf{B}$ is an order isomorphism. 
  
  The converse, i.e., the fact that any Bohr symmetry  is induced by a Jordan symmetry in the said way, is very deep (Hamhalter, 2011); see Theorem \ref{Hamhalter1} below. 
  \end{enumerate}
In view of these equivalences and Wigner's Theorem, we may conclude:
\begin{theorem}\L{BigTheorem}
Let $H$ be a \Hs, with $\dim(H)>2$ in nos.\ 5 and 6.
\begin{enumerate}
\item Each Wigner symmetry takes the form $\mathsf{W}(e)= ueu^*$ ($e\in\CP_1(H)$);
\item Each Kadison symmetry takes the form $\mathsf{K}(\rh)= u\rh u^*$ ($\rh\in\mathcal{D}(H)$); 
\item Each Jordan symmetry  takes the form $\mathsf{J}(a)=uau^*$; ($a\in B(H)_{\mathrm{sa}}$);
\item Each Ludwig symmetry  takes the form $\mathsf{L}(a)= uau^*$ ($a\in\CE(H)$); 
\item Each von Neumann symmetry  takes the form $\mathsf{N}(e)= ueu^*$ ($e\in\CP(H)$);
\item Each Bohr  symmetry  takes the form $\mathsf{B}(C)= uCu^*$ ($C\in  \CC(B(H))$),
\end{enumerate}
where in all cases the operator $u$ is either unitary or anti-unitary, and is uniquely determined by the symmetry in question ``up to a phase'' (that is, $u$ and  $u'$ implement the same symmetry by conjugation iff  $u'=zu$, where $z\in\T$). 
\end{theorem}
Of these six results, only the first and the third seem to have a direct proof; see e.g.\ Simon (1976) and Bratteli \&  Robinson (1987), respectively. Neither of these proofs is particularly elegant, so that especially a direct proof of no.\ 6 would be welcome. 
\section{Symmetries in algebraic quantum theory}\L{secA}
In this section we generalize the above analysis from $A=B(H)$ to arbitrary $C^*$-algebras $A$, which for simplicity we assume to have a unit $1_A$. 
\begin{enumerate}
\item  The \hi{pure state space} $P(A)=\partial_e S(A)$ of $A$ is the extreme boundary of the state space $S(A)$, 
seen as a uniform space equipped with a transition probability
\begin{equation}
\ta(\om,\om')= \inf\{\om(a)\mid a\in A, 0\leq a\leq 1_A,\om'(a)=1\}.\L{tpCstarA0}
\end{equation}
 If $A=B(H)$ and $\om,\om'$ lie in the normal pure state space $P_n(B(H))$ of $B(H)$, a simple computation (Landsman, 1998) shows that the above expression reproduces the standard quantum-mechanical transition probabilities \er{taborn}, but  compared to this special case one novel aspect of $P(A)$ is that all pure states are now taken into account (as opposed to merely the \emph{normal} ones, which notion is undefined for general $C^*$-algebras anyway). Another is that in order to obtain the desired equivalence with other structures, the set $P(A)$ should carry a uniform structure, namely the $w^*$-uniformity inherited from $A^*$. Thus 
a \hi{Wigner symmetry} of $A$ is a uniformly continuous bijection  $\mathsf{W}: P(A)\raw P(A)$ 
with  uniformly continuous inverse that preserves transition probabilities, i.e., that satisfies
\begin{equation}
\ta(\mathsf{W}(\om)\mathsf{W}(\om'))=\ta(\om,\om'), \:\: \om,\om'\in P(A).\L{WPA}
\end{equation}
\item The  \hi{state space} $S(A)$ is the set of all states on $A$, seen as a compact convex set in the $w^*$-topology inherited from the embedding $S(A)\subset A^*$. Hence 
a \hi{Kadison symmetry} of $A$ is an affine homeomorphism  $\mathsf{K}: S(A)\raw S(A)$.
Compared to the case $A=B(H)$, firstly \emph{all} states are now taken into account (instead of all \emph{normal} states), and secondly we have added a continuity condition on $\mathsf{K}$.
\item Any $C^*$-algebra $A$ defines an associated \hi{Jordan algebra}--more precisely, a $JB$-algebra if the norm is taken into account, cf.\ Hanche-Olsen \& St\o rmer (1984)--namely $A_{\mathrm{sa}}$ equipped with the commutative product $a\circ b=\half(ab+ba)$. A
 \hi{Jordan symmetry} $\mathsf{J}$  of $A$ is a Jordan isomorphism of $(A_{\mathrm{sa}},\circ)$
 (equivalently, a unital linear order isomorphism of 
 $(A_{\mathrm{sa}},\leq)$,  cf.\ 
 Alfsen \& Shultz (2001), Prop.\ 4.19). 
  \item The \hi{effects} in $A$ comprise the order unit interval $\CE(A)=[0,1_A]$, i.e., the set of all $a\in A_{\mathrm{sa}}$
  such that $0\leq a\leq 1_A$, seen as a convex poset as for $B(H)$. Hence a
  \hi{Ludwig symmetry} of $A$ is an affine order  isomorphism $\mathsf{L}: \CE(A)\raw\CE(A)$.
 \item The projections  $\CP(A)$ in $A$ form an orthocomplemented poset  with $e\leq f$ iff $ef=e$ and $e^{\perp}=1_A-e$;
  if $A$ is a \vna\ or more generally an $AW^*$-algebra or a Rickart \ca,   $\CP(A)$ is even an orthocomplemented lattice.
A \hi{von Neumann symmetry} of $A$ is an  invertible map  $\mathsf{N}: \CP(A)\raw\CP(A)$
 that preserves 0 and $\perp$ (and hence preserves 1) and satisfies $\phv(x\vee y)=\phv(x)\vee\phv(y)$ if 
$x\leq y^{\perp}$ (in which case $x\vee y$ is defined, as is always the case if $\CP(A)$ is a lattice). 
\item The poset $\CC(A)$ lying at the heart of exact Bohrification consists of all \emph{commutative} $C^*$-subalgebras of $A$ that contain the unit $1_A$, 
 partially ordered  by  inclusion. 
 A \hi{Bohr symmetry} of $A$, then,  is an order isomorphism $\mathsf{B}: \CC(A)\raw \CC(A)$.
\end{enumerate}

The structures 1, 2, 3, and 4 are equivalent, as follows.
\begin{theorem}\L{Shultz}
 Let $A$ and $B$ be unital $C^*$-algebras. The relation $f=\phv^*$ (that is, $f(\om)(a)=\om(\phv(a))$, where $a\in A_{\mathrm{sa}}$ and $\om\in P(A)$ or $\om\in S(A)$, and $\phv$ and $f$ are specified below) gives a bijective correspondence
 between:
\begin{enumerate}
\item Jordan isomorphisms $\phv:A_{\mathrm{sa}}\raw B_{\mathrm{sa}}$;
\item  Bijections
$f:P(B)\raw P(A)$ that preserve transition probabilities and are $w^*$-uniformly continuous along with their inverse;
\item  Affine homeomorphisms $f:S(B)\raw S(A)$. 
\end{enumerate}
\end{theorem}
 See Shultz (1982) for $1\lraw 2$ in this theorem and see Alfsen \&\ Shultz (2001),  Corollary 4.20, for $1\lraw 3$. The first of these
 equivalences also follows from the reconstruction of $(A_{\mathrm{sa}},\circ)$ from $P(A)$ in Landsman (1998), whereas
 the second  follows from the reconstruction of $(A_{\mathrm{sa}},\circ)$ from $S(A)$ 
 in Alfsen \&\ Shultz (2003).
The equivalence $3\lraw 4$ on the above list 1--6 is proved in  the same way as for $A=B(H)$.

The case of projections is more complicated, since many $C^*$-algebras have few projections (think of $A=C([0,1])$). Therefore, the poset $\CP(A)$ of all projections in $A$ can only be as informative as the four invariants just discussed under special assumptions on $A$. In the absence of a general result, we single out the class of $AW^*$-algebras 
as a particularly nice one in so far as abundance of projections is concerned. Recall that a $C^*$-algebra $A$ is an $AW^*$-algebra if for each nonempty subset $S\subseteq A$ there is a projection $e\in \CP(A)$ so that $R(S)=eA$, where the \emph{right-annihilator} $R(S)$ of $S\subseteq A$ is defined
as  $R(S)=\{a\in A\mid ba=0\, \forall b\in S\}$, and $R(a)\equiv R(\{a\})$. It follows that if it exists, $e$ is uniquely determined by $S$. For example, all \vna s are $AW^*$-algebras, so this class is vast. See 
 Berberian (1972). 
 
 The key result of interest to our theme is then provided by
Hamhalter's generalization of  Dye's Theorem to  $AW^*$-algebras (Hamhalter, 2015):
\begin{theorem}\L{HamHeuRey}
Let $A$ and $B$ be $AW^*$-algebras and let 
$\mathsf{N}:\CP(A)\raw\CP(B)$ be an isomorphism of the corresponding orthocomplemented projection lattices  that in addition preserves arbitrary suprema. If $A$ has no summand isomorphic to either $\C^2$ or $M_2(\C)$, then there is a unique Jordan isomorphism 
$\mathsf{J}: A_{\mathrm{sa}}\raw B_{\mathrm{sa}}$ that extends $\mathsf{N}$ (and hence Jordan isomorphisms are characterized by their values on projections). 
\end{theorem}
We omit the proof, as it is not related to our main topic of interest $\CA$; the proof follows from a theorem of Heunen and Reyes (2014) on projections in  $AW^*$-algebras and Hamhalter's own Gleason's Theorem for homogeneous $AW^*$-algebras.

We now move to the posets $\CA$.
Since the  structures 1--4 are equivalent, we may pick the one that is most convenient for a comparison with $\CA$, which turns out to be the Jordan algebra structure of $A$.
Henceforth $A$ and $B$ are unital $C^*$-algebras, and
 we define a \hi{weak Jordan isomorphism}{, also called a \emph{quasi Jordan isomorphism},} of $A$ and $B$ as  an invertible map 
 $\mathsf{J}:A_{\mathrm{sa}}\raw B_{\mathrm{sa}}$ whose restriction to  each subspace $C_{\mathrm{sa}}$ of $A_{\mathrm{sa}}$, where $C\in\CC(A)$, is linear and 
preserves the Jordan product $\circ$ (so that a Jordan symmetry of $A$ alone is a weak Jordan automorphism of of $A$).
Such a map complexifies to a map
$\mathsf{J}_{\C}:A\raw B$ in the same way as for $A=B=B(H)$. If no confusion arises, we  write $\mathsf{J}$ for $\mathsf{J}_{\C}$. 
\begin{proposition}
Given a weak Jordan isomorphism  $\mathsf{J}:A_{\mathrm{sa}}\raw B_{\mathrm{sa}}$, the ensuing map 
$\mathsf{B}: \CC(A)\raw\CC(B)$ defined by $\mathsf{B}(C)=\mathsf{J}_{\C}(C)$ is an order isomorphism.
\end{proposition}
Note that as an argument of $\mathsf{B}$ the symbol $C$ is a point in the poset $\CA$, whereas as an argument of $\mathsf{J}_{\C}$ it is a subset of $A$, so that $\mathsf{J}_{\C}$ stands for $\{\mathsf{J}_{\C}(c)\mid c\in C\}$.
The proof is elementary and is practically the same as for the special case $A=B=B(H)$; see also Hamhalter (2011), Proposition 1.1.
\section{Hamhalter's Theorem}\L{secHam}
The converse, however, is a deep result, due to Hamhalter (2011), Theorem 3.4.
\begin{theorem}\L{Hamhalter1}
Let $A$ and $B$ be unital $C^*$-algebras and let $\mathsf{B}:  \CC(A)\raw\CC(B)$ be an order isomorphism. Then there is a weak Jordan isomorphism $\mathsf{J}:A_{\mathrm{sa}}\raw B_{\mathrm{sa}}$ such that { $\mathsf{B}(C)=\mathsf{J}_{\C}[C]$ for each $C\in\CC(A)$}. Moreover, if $A$ is isomorphic to neither $\C^2$ nor $M_2(\C)$, then $\mathsf{J}$ is uniquely determined by $\mathsf{B}$, so in that case there is a bijective correspondence $\mathsf{J} \lraw\mathsf{B}$ between weak Jordan symmetries  $\mathsf{J}$
of $A$ and Bohr symmetries $\mathsf{B}$ of $A$. 
\end{theorem}
The question whether the weak Jordan isomorphism in question is a Jordan isomorphism will be postponed to Theorem \ref{Hamhalter2} below. 

Before proving Theorem \ref{Hamhalter1}, let us explain why $\C^2$ and $M_2(\C)$ are exceptional. 
\begin{itemize}
\item The only  order isomorphism of the
poset $\CC(\C^2)\cong\{0,1\}$ (with $0\equiv \C\cdot 1_2$ and $1\equiv\C^2$)
is the identity map, which is induced by both the map $(a,b)\mapsto (b,a)$
and by the identity map on $\C^2$ (each of which is a weak Jordan automorphism).
\item  The poset 
$\CC(M_2(\C))$ has a bottom element $0\equiv  \C\cdot 1_2$, as before, but no top element; each element 
$C\neq\C\cdot 1_2$
of $\CC(M_2(\C))$ is a unitary conjugate of the diagonal subalgebra $D_2(\C)$, with $0\leq C$ but no other orderings. 
Furthermore, $C\cap \til{C}=\C\cdot 1_2$ whenever $C\neq \til{C}$. Hence any order isomorphism of $\CC(M_2(\C))$ maps $\C\cdot 1_2$ to itself and permutes the $C$'s. Thus each map $\mathsf{J}: M_2(\C)_{\mathrm{sa}}\raw M_2(\C)_{\mathrm{sa}}$ whose complexification $\mathsf{J}_{\C}:M_2(\C)\raw M_2(\C)$ shuffles the $C$'s isomorphically (as $C^*$-algebras)
gives a weak Jordan automorphism. For example, take $(a,b)\mapsto (b,a)$ on $D_2(\C)$ and the identity on each $C\neq D_2(\C)$; this  induces the identity map on $\CC(M_2(\C))$. It follows that there are vastly more weak Jordan automorphisms of $M_2(\C)$ than there are order isomorphisms of  $\CC(M_2(\C))$.
\end{itemize}
The proof of Theorem \ref{Hamhalter1} deserves a section on its own; we roughly follow Hamhalter (2011), but add various details and also take some different turns.
The main differences with the original proof by Hamhalter are the following. Firstly, we give an order-theoretic characterization of u.s.c. decompositions of the form $\pi_K$ (and hence of algebras in $\CC(C(X))$ that are the unitization of some ideal) by three axioms as stated in Lemma 3.1.1 in Firby (1973), whereas Hamhalter uses Proposition 7 in Mendivil (1999), which gives a different characterization of unitizations of ideals.
	Furthermore, Hamhalter only treats Lemma \ref{Hamlemma} in full generality, whereas in our opinion it is very instructive to take the case of finite sets first,
	where many of the key ideas already appear in a setting where they are not overshadowed by topological complications.
	Finally, our proof of Lemma \ref{Hamlem2}.2 differs from Hamhalter's proof. 
\begin{proof}
The key to the proof lies in the commutative case, which can be reduced to topology. If $A=C(X)$,
any $C\in\CA$ induces an equivalence relation $\sim_C$ on $X$ by
\begin{equation}
x\sim_C y \mbox{ iff } f(x)=f(y)\: \forall\, f\in C.\L{simC}
\end{equation}
This, in turn, defines a partition $X=\bigsqcup_{\lm} K_{\lm}$ of $X$ (henceforth called $\pi$), whose blocks  $K_{\lm}\subset X$ are the equivalence classes of $\sim_C$. To study a possible inverse of this procedure, for any closed subset $K\subset X$ we define the ideal 
\beq 
I_K=C(X;K)=\{f\in C(X)\mid f(x)=0\,\forall\, x\in K\},
\eeq
in $C(X)$, 
 and its unitization $\dot{I}_K=I_K\oplus\C\cdot 1_X$, which evidently consists of all continuous functions on $X$ that are constant on $K$. 
 If $X$ is finite (and discrete), 
each partition $\pi$ of $X$  defines some unital  $C^*$-algebra $C\subseteq C(X)$ through
\begin{equation}
C=\bigcap_{K_{\lm}\in\pi}\dot{I}_{K_{\lm}}, \L{Cbigcap}
\end{equation}
which consists of all $f\in C(X)$ that are constant on each block $K_{\lm}$ of the given partition $\pi$. 
In that case, the correspondence $C\lraw\pi$, where $\pi$ is defined by the equivalence relation $\sim_C$ in \er{simC},
gives a bijection between $\CC(C(X))$ and the set $\mathfrak{P}(X)$ of all partitions of $X$. For example, the subalgebra 
$C=\dot{I}_K$ corresponds to the partition consisting of $K$ and all singletons not lying in $K$.
 Given the already defined partial order on $\CC(C(X))$ (i.e., $C\leq D$ iff $C\subseteq D$), we may promote this bijection to an order isomorphism of posets if we define a partial order on $\mathfrak{P}(X)$ to be the \emph{opposite} of the usual one in which $\pi\leq \pi'$ (where $\pi$ and $\pi'$ consist of blocks $\{K_{\lm}\}$ and
 $\{K'_{\lm'}\}$, respectively) iff  each $K_{\lm}$ is contained in some $K'_{\lm'}$ (i.e., $\pi$ is finer than $\pi'$).
 This partial ordering makes  $\mathfrak{P}(X)$ a complete lattice, whose bottom element consists of all singletons on $X$ and whose top element just consists of $X$ itself: the former corresponds to $C(X)$, which is the top element of $\CC(C(X))$, whilst the latter  corresponds to $\C\cdot 1_X$, which is the bottom element of  $\CC(C(X))$. 

For general compact Hausdorff spaces $X$, since $C(X)$ is sensitive to the topology of $X$ the equivalence relation \er{simC} does not induce arbitrary partitions of $X$. It turns out  that each 
$C\in\CC(C(X))$ induces an  \hi{upper semicontinuous partition} (abbreviated by \emph{u.s.c. decomposition}) of $X$, i.e.,
\begin{itemize}
\item Each block $K_{\lm}$ of  the partition $\pi$  is closed;
\item For each block $K_{\lm}$ of $\pi$, if $K_{\lm}\subseteq U$ for some open $U\in\CO(X)$, then there is $V\in\CO(X)$ such that $K_{\lm}\subseteq V \subseteq U$
and $V$ is a union of blocks of $\pi$ (in other words, if $K$ is such a block, then $V\cap K=\emptyset$ implies
{$K=\emptyset$}). 
\end{itemize}
This can be seen as follows. Firstly, if we equip $\pi$ with the quotient topology with respect to the the natural map $q:X\to\pi$, $x\mapsto K_\lambda$ if $x\in K_\lambda$, then $\pi$ is compact, for $X$ is compact. Moreover, $\pi$ is Hausdorff: let $K_\lambda$ and $K_\mu$ be two distinct \emph{points} in $\pi$. Recall that $x,y\in K_\lambda$ if and only if $f(x)=f(y)$ for each $f\in C$. Since $K_\lambda\neq K_\mu$, there is some $x\in K_\lambda$, some $y\in K_\mu$ and some $f\in C$ such that $f(x)\neq f(y)$, whence there are open disjoint $U,V\subseteq\C$ such that $f(x)\in U$ and $f(y)\in V$. 

Define $\hat f:\pi\to\C$ by $\hat f(K_\lambda)=f(x)$ for some $x\in K_\lambda$. 
By definition of $K_\lambda$ this is independent of the choice of $x\in K_\lambda$, hence $\hat f$ is well defined. Again by definition, we have $f=\hat f\circ q$, hence $q^{-1}(\hat f^{-1})[U]=f^{-1}[U]$, which is open in $X$ since $f$ is continuous. Since $\pi$ is equipped with the quotient topology, it follows that $\hat f^{-1}[U]$ is open in $\pi$, and similarly $\hat f^{-1}[V]$ is open. Moreover, we have $\hat f(K_\lambda)=f(x)$ and $f(x)\in U$, hence $K_\lambda\in \hat f^{-1}[U]$, and similarly, $K_\mu\in\hat f^{-1}[V]$. We conclude that $\pi$ is also Hausdorff. Since $q$ is a continuous map between compact Hausdorff spaces, it follows that $q$ is closed.
It now follows from Theorem 9.9 in Willard (1970)---which also gives further background on  decompositions---that $\pi$ is a u.s.c. decomposition.

Consequently, by the same maps \er{simC} and \er{Cbigcap}, the poset $\CC(C(X))$ is anti-isomorphic to  the poset  
$\mathfrak{F}(X)$  of all u.s.c. decompositions of $X$ (this proves that $\mathfrak{F}(X)$ is a complete lattice, since $\CC(C(X))$ is).  This is still a complicated poset; assuming $X$ to be larger than a singleton, the next step is to identify the simpler poset  $\mathcal{F}_2(X)$ of all closed subsets of $X$ containing at least two elements within $\mathfrak{F}(X)$, where (as above) we identify a closed $K\subseteq X$ with the (u.s.c.) partition $\pi_K$ of $X$ whose blocks are $K$ and all singletons not lying in $K$ (note that the poset $\mathcal{F}(X)$ of all closed subsets of $X$ is less useful, since any singleton in $\mathcal{F}(X)$ gives rise to the bottom element of $\mathfrak{F}(X)$).
To do so, we first recall that $\beta$ is said to \emph{cover} $\alpha$ in some poset if $\alpha<\beta$, and $\alpha\leq\gamma<\beta$ implies $\alpha=\gamma$. If the poset has a bottom element, then its covers are called \emph{atoms}. Furthermore,
note that since the bottom element $0$ of $\mathfrak{F}(X)$ consists of singletons, the atoms in $\mathfrak{F}(X)$ are the partitions of the form $\pi_{\{x_1,x_2\}}$  (where $x_1\neq x_2$). It follows that some partition $\pi\in\mathfrak{F}(X)$ lies in $\mathcal{F}_2(X)\subset \mathfrak{F}(X)$ iff exactly
one of the following conditions holds:
\begin{itemize}
\item $\pi$ is an atom in $\mathfrak{F}(X)$, i.e., $\pi=\pi_{\{x_1,x_2\}}$ for some $x_1,x_2\in X$, $x_1\neq x_2$; 
\item  $\pi$ covers  three (distinct) atoms in $\mathfrak{F}(X)$, in which case $\pi=\pi_{\{x_1,x_2,x_3\}}$ where all $x_i$ are different, which covers the atoms $\pi_{\{x_1,x_2\}}$, $\pi_{\{x_1,x_3\}}$, and $\pi_{\{x_2,x_3\}}$;
\item If $\al\neq \beta$ are  atoms  in $\mathfrak{F}(X)$ such that $\al\leq \pi$ and $\beta\leq\pi$, there is an atom $\gm\leq\pi$ such that there are three  (distinct) atoms covered by $\al\vee \gm$ and three  (distinct) atoms covered by $\beta\vee \gm$.
In that case, $\pi=\pi_K$ where $K$ has more than three elements: if $\al=\pi_{\{x_1,x_2\}}$ and $\beta=\pi_{\{x_3,x_4\}}$, then due to the assumption $\al\neq \beta$, the set $\{x_1,x_2,x_3,x_4\}$ (which lies in $K$) has at least three distinct elements,  say $\{x_1,x_2,x_3\}$. 
Hence we may take $\gm=\pi_{\{x_2,x_3\}}$, in which case $\al\vee\gm=\pi_{\{x_1,x_2,x_3\}}$, which covers the atoms $\al$, $\gm$, and $\pi_{\{x_1,x_3\}}$. Likewise, we have $\beta\vee\gm=\pi_{\{x_2,x_3,x_4\}}$, which covers three atoms  $\beta$, $\gm$, and $\pi_{\{x_2,x_4\}}$.
\end{itemize}
In order to see that $\pi$ satisfying the third condition must be of the form $\pi_K$, assume the converse. So $\pi$ contains two blocks $K_\lambda$ and $K_\mu$ consisting of two or more elements. Say $\{x_1,x_2\}\subseteq K_\lambda$ and $\{x_3,x_4\}\subseteq K_\mu$. Then $\alpha=\pi_{\{x_1,x_2\}}$ and $\beta_{\{x_3,x_4\}}$ are atoms such that $\alpha,\beta<\pi$, and there is an atom $\gamma=\pi_{\{x_5,x_6\}}\leq\pi$ such that there are three atoms covered by $\alpha\vee\gamma$, and there are three atoms covered by $\beta\vee\gamma$. It follows from the second condition that $\alpha\vee\gm=\pi_L$ with $L$ a three-point set. This implies that $\{x_1,x_2\}\cap\{x_5,x_6\}$ is not empty, from which it follows that $\alpha\vee\gm=\pi_{\{x_1,x_2,x_5,x_6\}}$. Similarly, we find $\beta\vee\gamma=\pi_{\{x_3,x_4,x_5,x_6\}}$. Since $\{x_1,x_2,x_5,x_6\}$ and $\{x_3,x_4,x_5,x_6\}$ overlap, we obtain $\alpha\vee\beta\vee\gamma=\pi_{\{x_1,x_2,x_3,x_4,x_5,x_6\}}$. Moreover, $\alpha,\beta,\gamma\leq\pi$, so $\alpha\vee\beta\vee\gamma\leq\pi$. However, since $x_1,x_2\in K_\lambda$, we must have $\{x_1,x_2,x_3,x_4,x_5,x_6\}\subseteq K_\lambda$ by definition of the order on $\mathfrak{F}(X)$. But since $x_3,x_4\in K_\mu$, we must also have $\{x_1,x_2,x_3,x_4,x_5,x_6\}\subseteq K_\mu$, which is not possible, since $K_\lambda$ and $K_\mu$ are distinct blocks, hence disjoint. We conclude that $\pi$ can have only one block $K$ of two or more elements, hence $\pi=\pi_K$.     

Thus  $\mathcal{F}_2(X)\subset \mathfrak{F}(X)$ has been characterized  order-theoretically. Moreover, 
\begin{equation}
\pi=\vee_{x\in X} \pi_{K(x)},\L{piveex}
\end{equation}
where $K(x)$ is the unique block of $X$ that contains $x$. Hence $\mathcal{F}_2(X)$ determines $\mathfrak{F}(X)$.

Let $X$ and $Y$ be compact Hausdorff spaces of cardinality at least two (so that the empty set and singletons are excluded). 
By the previous analysis, an order isomorphism $\mathsf{B}:\CC(C(X))\raw\CC(C(Y))$ is equivalent to an 
order isomorphism $\mathfrak{F}(X)\raw \mathfrak{F}(Y)$, which in turn restricts to an order isomorphism $\mathcal{F}_2(X)\raw \mathcal{F}_2(Y)$. 
\begin{lemma}\L{Hamlemma}
If $X$ and $Y$ are compact Hausdorff spaces of cardinality at least two, then any order isomorphism $\mathsf{F}:\mathcal{F}_2(X)\raw \mathcal{F}_2(Y)$ is induced by a homeomorphism $\phv:X\raw Y$ via $\mathsf{F}(F)=\phv(F)$, i.e.,
$\mathsf{F}(F)=\cup_{x\in F}\{\phv(x)\}$.
Moreover,
if  $X$ and $Y$ have  cardinality at least three, then $\phv$ is uniquely determined by $\mathsf{F}$.
\end{lemma}
We first prove this for finite $X$, where  $\mathcal{F}_2(X)$ simply consists of all subsets of $X$ having at least two elements, etc. It is easy to see that $X$ and $Y$ must have the same cardinality $|X|=|Y|=n$. If $n=2$, then  $\mathcal{F}_2(X)=X$ etc., so there is only one map $\mathsf{F}$, which  is induced by each of the two possible maps $\phv:X\raw Y$, so that $\phv$ exists but fails to be unique. 
If $n>2$, then $\mathsf{F}$ must map each subset of $X$ with $n-1$ elements to some subset of $Y$ with $n-1$ elements, so that taking complements we obtain a unique bijection $\phv:X\raw Y$. To show that $\phv$ induces $\mathsf{F}$,  note that the meet  $\wed$ in $\mathcal{F}_2(X)$ is simply intersection $\cap$, and also that for any $F\in \mathcal{F}_2(X)$,
$$F=\cup_{x\in F}\{x\}=\cap_{x\notin F}\{x\}^c=(\cup_{x\notin F}\{x\})^c,$$
 where $A^c=X\backslash A$. Since $\mathsf{F}$ is an order isomorphism it preserves $\wed=\cap$, so that
$$
\mathsf{F}(F)=\cap_{x\notin F}\mathsf{F}(\{x\}^c)=\cap_{x\notin F}X\backslash
 \{\phv(x)\}=(\cup_{x\notin F}\{\phv(x)\})^c=\cup_{x\in F}\{\phv(x)\}.
$$

Now assume that $X$ is infinite. Let $x\in X$. If $x$ is not isolated, we define $ \phv(x)$ as follows. Let $\mathcal{O}(x)$ denote the set of all open neighborhoods of $x$. Since $x$ is not isolated, each $O\in\mathcal{O}(x)$ contains at least another element, so $\overline{O}\in\mathcal{F}_2(X)$. Moreover, finite intersections of elements of $\{\overline{O}:O\in\mathcal{O}(x)\}$ are still in $\mathcal{F}_2(X)$. Indeed, if $O_1,\ldots,O_n\in\mathcal{O}(x)$, then $O_1\cap\ldots\cap O_n$ is an open set containing $x$, and since $\overline{O_1\cap\ldots\cap O_n}\subseteq\overline{O_1}\cap\ldots\cap\overline{O_n}$, it follows that $\overline{O_1}\cap\ldots\cap\overline{O_n}\in\mathcal{F}_2(X)$. Since $\mathsf{F}$ is an order isomorphism, we find that finite intersections of $\{\mathsf{F}(\overline{O}):O\in\mathcal{O}(x)\}$ are contained in $\mathcal{F}_2(Y)$. This implies that $\{\mathsf{F}(\overline{O}):O\in\mathcal{O}(x)\}$ satisfies the finite intersection property. As $Y$ is compact, it follows that $I_x=\bigcap_{O\in\mathcal{O}(x)}\mathsf{F}(\overline{O})$ is non-empty. We can say more: it turns out that $I_x$ contains exactly one element. Indeed, assume that there are two different points $y_1,y_2\in I_x$. Then $\{y_1,y_2\}\in\mathcal{F}_2(Y)$, so $\mathsf{F}^{-1}(\{y_1,y_2\})\in\mathcal{F}_2(X)$. Since $\{y_1,y_2\}\in\mathsf{F}(\overline{O})$ for each $O\in\mathcal{O}(x)$, we also find that $\mathsf{F}^{-1}(\{y_1,y_2\})\subseteq\overline{O}$ for each $O\in\mathcal{O}(x)$. This implies that $$\mathsf{F}^{-1}(\{y_1,y_2\})\subseteq\bigcap_{O\in\mathcal{O}(x)}\overline{O}=\{x\},$$ where the last equality holds by  normality of $X$. But this is a contradiction with $\mathsf{F}:\mathcal{F}_2(X)\to\mathcal{F}_2(Y)$ being a bijection. So $I_x$ contains exactly one point. We define $ \phv(x)$ such that $\{ \phv(x)\}=I_x$. Notice that $ \phv(x)$ cannot be isolated in $Y$, since if we assume otherwise, then $Y\setminus\{ \phv(x)\}$ must be a co-atom in $\mathcal{F}_2(Y)$, whence $\mathsf{F}^{-1}(Y\setminus\{ \phv(x)\})$ is a co-atom in $\mathcal{F}_2(X)$, which must be of the form $X\setminus\{z\}$ for some isolated $z\in X$. Since $x$ is not isolated, we cannot have $x=z$, so $X\setminus\{z\}$ is an open neighborhood of $x$, which is even clopen since $z$ is isolated. By definition of $ \phv(x)$, we must have $\phv(x)\in\mathsf{F}(X\setminus\{z\})$, but $\mathsf{F}(X\setminus\{z\})=Y\setminus\{\phv(x)\}$. We found a contradiction, hence $ \phv(x)$ cannot be isolated.
Now assume that $x$ is an isolated point. Then $X\setminus\{x\}$ is a co-atom in $\mathcal{F}_2(X)$, so $\mathsf{F}(X\setminus\{x\})$ is a co-atom in $\mathcal{F}_2(Y)$, too. Clearly this implies that $\mathsf{F}(X\setminus\{x\})=Y\setminus\{y\}$ for some unique $y\in Y$, which must be isolated, since $Y\setminus\{y\}$ is closed. We define $ \phv(x)=y$.

In an analogous way, $\mathsf{F}^{-1}$ induces a map $\psi:Y\to X$. We shall show that $\phv$ and $\psi$ are each other's inverses. Let $x\in X$ be isolated. We have seen that $\phv(x)$ must be isolated as well, and that $\phv(x)$ is defined by the equation $\mathsf{F}(X\setminus\{x\})=Y\setminus\{ \phv(x)\}$. Since $\mathsf{F}$ is an order isomorphism, we have $X\setminus\{x\}=\mathsf{F}^{-1}(Y\setminus\{ \phv(x)\})$. Since $ \phv(x)$ is isolated, we find by definition of $\psi$ that $\psi( \phv(x))=x$. In a similar way we find that $ \phv(\psi(y))=y$ for each isolated $y\in Y$.
Now assume that $x$ is not isolated and let $F\in\mathcal{F}_2(X)$ such that $x\in F$. Then
\begin{eqnarray*}
	\{ \phv(x)\} & = & \bigcap_{O\in\mathcal{O}(x)}\mathsf{F}(\overline{O})
	 \subseteq  \bigcap\{\mathsf{F}(\overline{O}):O\ \mathrm{open},F\subseteq O\}\\
	& =& \mathsf{F}\left(\bigcap\{\overline{O}:O\ \mathrm{open},F\subseteq O\}\right)
	 =  \mathsf{F}(F),
\end{eqnarray*}
where the last equality follows by completely regularity of $X$. The penultimate equality follows from the following facts. Firstly, the set $\bigcap\{\overline {O}:O\ \mathrm{open},F\subseteq O\}$ is closed since it is the intersection of closed sets. Moreover, the intersection contains more than one point, since $F$ contains two or more points and $F\subseteq\overline{O}$ for each $O$. Hence $\bigcap\{\overline {O}:O\ \mathrm{open},F\subseteq O\}\in\mathcal{F}_2(X)$, and since $\mathsf{F}$ is an order isomorphism, it preserves infima, which justifies the penultimate equality. Hence $ \phv(x)\in\mathsf{F}(F)$ for each $F\in\mathcal{F}_2(X)$ containing $x$. Since $x$ is not isolated, $ \phv(x)$ is not isolated either. Hence in a similar way, we find that $\psi( \phv(x))\in\mathsf{F}^{-1}(G)$ for each $G\in\mathcal{F}_2(Y)$ containing $ \phv(x)$. Let $z=\psi( \phv(x)$. Combining both statements, we find that $z\in F$ for each $F\in\mathcal{F}_2(X)$ such that $x\in F$. In other words, $z\in\bigcap\{F\in\mathcal{F}_2(X):x\in F\}$. Since $x$ is not isolated, we each $O\in\mathcal{O}(x)$ contains at least two points. Hence $$\bigcap\{F\in\mathcal{F}_2(X):x\in F\}\subseteq\bigcap\{\overline{O}:O\in\mathcal{O}(x)\}=\{x\},$$ where we used complete regularity of $X$ in the last equality. We conclude that $z=x$, so $\psi( \phv(x))=x$. In a similar way, we find that $ \phv(\psi(y))=y$ for each non-isolated $y\in Y$. We conclude that $\phv$ is a bijection with $\phv^{-1}=\psi$.

We have to show that if $F\in\mathcal{F}_2(X)$, then $\phv[F]=\mathsf{F}(F)$. Let $x\in F$. When we proved that $\phv$ is a bijection, we already noticed that $\phv(x)\in\mathsf{F}(F)$ if $x$ is not isolated. If $x$ is isolated in $X$, then we first assume that $F$ has at least three points. Since $\{x\}$ is open, $G=F\setminus\{x\}$ is closed. Since $F$ contains at least three points, $G\in\mathcal{F}_2(X)$. So $G$ is covered by $F$ in $\mathcal{F}_2(X)$, so $\mathsf{F}(F)$ covers $\mathsf{F}(G)$. It follows that there must be an element $y_G\in Y\setminus\mathsf{F}(G)$ such that $$\mathsf{F}(F)=\mathsf{F}(G\cup\{x\})=\mathsf{F}(G)\cup\{y_G\}.$$
We have $G\cup\{x\},X\setminus\{x\}\in\mathcal{F}_2(X)$, so
\begin{eqnarray*}
	\mathsf{F}(G) & = & \mathsf{F}(G\cup\{x\}\cap X\setminus\{x\})=\mathsf{F}(G\cup\{x\})\cap\mathsf{F}(X\setminus\{x\})\\
	& = & (\mathsf{F}(G)\cup\{y_G\})\cap(Y\setminus\{ \phv(x)\}),
\end{eqnarray*}
where $\mathsf{F}(X\setminus\{x\})=Y\setminus\{ \phv(x)\}$ by definition of values of $\phv$ at isolated points. Since $x\notin G$ and $\mathsf{F}$ preserves inclusions, this latter equation also implies $\mathsf{F}(G)\subseteq Y\setminus\{ \phv(x)\}$. Hence we find $$\mathsf{F}(G)=(\mathsf{F}(G)\cup\{y_G\})\cap(Y\setminus\{ \phv(x)\})=\mathsf{F}(G)\cup(\{y_G\}\cap Y\setminus\{ \phv(x)\}).$$ Thus we obtain $\{y_G\}\cap Y\setminus\{ \phv(x)\}\subseteq\mathsf{F}(G)$, but since $y_G\notin\mathsf{F}(G)$, we must have $ \phv(x)=y_G$. As a consequence, we obtain $\mathsf{F}(F)=\mathsf{F}(G)\cup\{ \phv(x)\},$ so $ \phv(x)\in\mathsf{F}(F)$.

Summarizing, if $F$ has at least three points, then $ \phv(x)\in\mathsf{F}(F)$ for $x\in F$, regardless whether $x$ is isolated or not. So $\phv[F]\subseteq\mathsf{F}(F)$ for each $F\in\mathcal{F}_2(X)$ such that $F$ has at least three points. Let $F\in\mathcal{F}_2(X)$ have exactly two points. Then there are $F_1,F_2\in\mathcal{F}_2(X)$ with exactly three points such that $F=F_1\cap F_2$. Then since $\phv$ is a bijection and $\mathsf{F}$ as an order isomorphism both preserve intersections in $\mathcal{F}_2(X)$, we find $$\phv[F]=\phv[F_1\cap F_2]= \phv[F_1]\cap \phv[F_2]\subseteq\mathsf{F}(F_1)\cap\mathsf{F}(F_2)=\mathsf{F}(F_1\cap F_2)=\mathsf{F}(F).$$
So $\phv[F]\subseteq\mathsf{F}(F)$ for each $F\in\mathcal{F}_2(X)$. In a similar way, we find $\phv^{-1}[G]\subseteq\mathsf{F}^{-1}[G]$ for each $G\in\mathcal{F}_2(Y)$. So if we substitute $G=\mathsf{F}(F)$, we obtain $\phv^{-1}[\mathsf{F}(F)]\subseteq F$. Since $\phv$ is a bijection, it follows that $\mathsf{F}(F)=\phv[F]$ for each $F\in\mathcal{F}_2(X)$. As a consequence, $\phv$ induces a one-one correspondence between closed subsets of $X$ and closed subsets of $Y$. Hence $\phv$ is a homeomorphism.

This proves Lemma \ref{Hamlemma}. The special case of Theorem \ref{Hamhalter1} where $A$ and $B$ are commutative now follows if we combine all steps so far:
\begin{enumerate}
\item  The Gelfand isomorphism allows us to assume $A=C(X)$ and $B=C(Y)$, as above;
\item The order isomorphism  $\mathsf{B}: \CC(A)\raw\CC(B)$ determines and is determined by an 
 order isomorphism $F:\mathfrak{F}(X)\raw \mathfrak{F}(Y)$ of the underlying lattices of u.s.c. decompositions;
 \item Because of \er{piveex}, the   order isomorphism $F$
  in turn  determines and 
 is determined by an  order isomorphism $\mathsf{F}:\mathcal{F}_2(X)\raw \mathcal{F}_2(Y)$;
 \item Lemma \ref{Hamlemma} yields a homeomorphism $\phv:X\raw Y$  inducing $\mathsf{F}:\mathcal{F}_2(X)\raw \mathcal{F}_2(Y)$;
 \item The inverse pullback $(\phv\inv)^*: C(X)\raw C(Y)$ is an isomorphism of $C^*$-algebras, which 
(running backwards) reproduces the initial map $\mathsf{B}: \CC(C(X))\raw\CC(C(Y))$.
\end{enumerate}
Therefore, in the commutative case we apparently obtain rather more than a weak Jordan isomorphism $\mathsf{J}:A_{\mathrm{sa}}\raw B_{\mathrm{sa}}$; we even found an isomorphism $\mathsf{J}:A\raw B$ of $C^*$-algebras. However,
if $A$ and $B$ are commutative, the condition of linearity on each commutative $C^*$-subalgebra $C$ of $A$ includes $C=A$, so that (after complexification)  weak Jordan isomorphisms are the same as isomorphisms of $C^*$-algebras. 

We now turn to the general case, in which $A$ and $B$ are both noncommutative (the case where one, say $A$, is commutative but the other is not cannot occur, since $\CC(A)$ would be a complete lattice but $\CC(B)$ would not).
Let $D$ and $E$ be maximal abelian $C^*$-subalgebras of $A$, so that the corresponding elements of $\CA$ are maximal in the order-theoretic sense.  Given an order isomorphism  $\mathsf{B}: \CC(A)\raw\CC(B)$, we restrict the map $\mathsf{B}$ to the down-set $\daw\! D=\CC(D)$ in $\CA$ so as to obtain an order homomorphism $\mathsf{B}_{|D}: \CC(D)\raw \CC(B)$. 
The image of $\CC(D)$ under  $\mathsf{B}$ must have a maximal element (since $\mathsf{B}$ is an order isomorphism), and so there is a maximal commutative  $C^*$-subalgebra $\til{D}$ of $B$ such that $\mathsf{B}_{|D}: \CC(D)\raw\CC(\til{D})$ is an order isomorphism. Applying the previous result, we obtain an isomorphism $\mathsf{J}_D: D\raw \til{D}$ of commutative $C^*$-algebras that induces $\mathsf{B}_{|D}$. The same applies to $E$, so we also have an isomorphism $\mathsf{J}_E: E\raw \til{E}$ of commutative $C^*$-algebras that induces $\mathsf{B}_{|E}$. Let $C=D\cap E$, which lies in $\CA$. 
We now show that 
$\mathsf{J}_D$ and $\mathsf{J}_E$ coincide on $C$. There are three cases. \begin{enumerate}
\item $\dim(C)=1$. In that case $C=\C\cdot 1_A$ is the bottom element of $\CA$, so it must be sent to the  bottom element $\til{C}=\C\cdot 1_B$ of $\CC(B)$, whence the claim.
\item $\dim(C)=2$. This the hard case dealt with below.
 \item $\dim(C)>2$. This case is settled by the uniqueness claim in Lemma \ref{Hamlemma}.
\end{enumerate}
So assume $\dim(C)=2$. In that case, $C=C^*(e)$ for some proper projection $e\in\CP(A)$, which is equivalent to $C$ being an atom in $\CA$. Recall that all our $C^*$-algebras are unital, and that by assumption $C^*$-subalgebras share the unit of the ambient $C^*$-algebra, hence $C^*(e)$ contains the unit of $A$. Hence $\til{C}\equiv \mathsf{B}(C)=\mathsf{B}_{|D}(C)=\mathsf{B}_{|E}(C)$  is an atom in $\CC(B)$, which implies that $\til{C}=C^*(\til{e})$ for some projection $\til{e}\in\CP(B)$. If $\mathsf{J}_D(e)=\mathsf{J}_E(e)$ we are done, so we must exclude the case $\mathsf{J}_D(e)=\til{e}$, $\mathsf{J}_E(e)=1_B-\til{e}$. This analysis again requires a case distinction:  
\begin{eqnarray}
 \dim(eAe)&=&\dim(e^{\perp}Ae^{\perp})=1; \L{caseeAe1}\\
 \dim(eAe)&=&1, \:\: \dim(e^{\perp}Ae^{\perp})>1;  \L{caseeAe2}\\
 \dim(eAe)&>&1, \:\: \dim(e^{\perp}Ae^{\perp})>1,\L{caseeAe3}
\end{eqnarray}
where $e^{\perp}=1_A-e$.
Each of these cases is nontrivial, and we need another lemma.
\begin{lemma} \L{Hamlem2}
Let  $C\in\CA$ be maximal (i.e., $C\subset A$ is maximal abelian).
\begin{enumerate}
\item For each projection $e\in\CP(C)$ we have $\dim(eCe)=1$ iff
$\dim(eAe)=1$.
\item We have $\dim(C)=2$ iff either $A\cong\C^2$ or $A\cong M_2(\C)$. 
\end{enumerate}
\end{lemma}
\begin{proof}
For the first claim $\dim(eAe)=1$ clearly implies $\dim(eCe)=1$. For the converse implication, assume  \emph{ad absurdum} that $\dim(eAe)>1$, so that there is an $a\in A$
for which $eae\neq \lm\cdot e$ for any $\lm\in\C$.
 If also  $\dim(eCe)=1$, then any $c\in C$ takes the form $c=\mu\cdot e+ e^{\perp}ce^\perp$ for some $\mu\in\C$. Indeed, since $c,e,e^\perp$ commute within $C$, 
  $$c=ce+ce^\perp=ce^2+c(e^\perp)^2=ece+e^\perp c e^\perp=\mu e+e^\perp c e^\perp,$$ where the last equality follows since $ece\in eCe$, which is spanned by $e$. This implies that $eae\in C'$ (where $C'$ is the commutant of $C$ within $A$), and since $C$ is maximal abelian, we have $C=C'$, whence $eae\in C$. Now $eae=e(eae)e$, hence $eae\in eCe$, whence $eae=\lambda\cdot e$ for some $\lambda\in\C$. 
Contradiction.
 
According to Exercise 4.6.12 in Kadison \& Ringrose (1981), showing that a $C^*$-algebra is finite-dimensional if it has finite-dimensional maximal abelian $\mbox{}^*$-subalgebra,
the assumption $\dim(C)=2$ implies that $A$ is finite-dimensional. The well-known theorem stating that every finite-dimensional $C^*$-algebra is isomorphic to a direct sum of matrix algebras then easily yields the second claim. 
\end{proof}

Having proved Lemma \ref{Hamlem2}, we move on the analyze the  cases
\er{caseeAe1} - \er{caseeAe3}. 
\begin{itemize}
\item Eq.\ \er{caseeAe1} implies that $C$ is maximal, as follows.  Any element $a\in A$ is a sum of
$eae$, $e^{\perp}ae^{\perp}$, $eae^{\perp}$, and $e^{\perp}ae$; nonzero elements of $C'=\{e\}'$ can only be of the first two types. If \er{caseeAe1} holds, then $\dim(C')=2$, but since $C$ is abelian we have $C\subseteq C'$ and since $\dim(C)=2$ we obtain $C'=C$.  Lemma \ref{Hamlem2}.2 then implies that either $A\cong\C^2$ or $A\cong M_2(\C)$. 
These $C^*$-algebras have been analyzed after the statement of Theorem \ref{Hamhalter1}, and since those two $A$'s conversely imply
\er{caseeAe1}, we may exclude them in dealing with \er{caseeAe2} - \er{caseeAe3}. By Lemma  \ref{Hamlem2}.2 (applied to $D$ and $E$ instead of $C$), in what follows we may assume that $\dim(D)>2$ and $\dim(E)>2$ (as $D$ and $E$ are maximal).  
\item Eq.\ \er{caseeAe2} implies $\dim(eD)=1$. Assuming $\mathsf{J}_D(e)=\til{e}$, this implies $\dim(\til{e}\til{D})=1$ (since $\mathsf{J}_D$ is an isomorphism). Applying
Lemma \ref{Hamlem2}.1 to $B$ gives $\dim(\til{e}B\til{e})=1$ (since $\til{D}$ is maximal). If also $\dim((1_B-\til{e})B(1_B-\til{e}))=1$,
then $\dim(\til{D})=2$, whence  $\dim(D)=2$, which we excluded. Hence $\dim((1_B-\til{e})B(1_B-\til{e}))>1$.
Applied to $\mathsf{J}_E$ this gives $\mathsf{J}_E(e)=\til{e}$,  and hence $\mathsf{J}_D$ and $\mathsf{J}_E$ coincide on $C=C^*(e)$.
\item Eq.\  \er{caseeAe3} implies that $\dim(eDe)>1$ as well as  $\dim(e^{\perp}Ee^{\perp})>1$ (apply Lemma \ref{Hamlem2}.1 to $D$ and $E$, respectively). Since $\dim(eDe)>1$, there is some $a\in D$ such that $e$ and $a'=eae\in D$ are linearly independent, and similarly there is some $b\in E$ such that $b'=e^\perp be^\perp$ is linearly independent from $e^\perp$. Then $a',b',e$ 
commute (in fact, $a'b'=b'a'=0$), so that we may form the abelian $C^*$-algebras
$C_1=C^*(e,a')\subseteq D$ and $C_2=C^*(e,b')\subseteq E$, which (also containing the unit $1_A$) both have dimension at least three. We also form  
 $C_3=C^*(e,a',b')$, which contains $C_1$ and $C_2$ and hence is at least three-dimensional, too.
 Because $D$ and $E$ are maximal abelian, $C_3$ must lie in both $D$ and $E$.   Applying the abelian case of the theorem already proved to $D$ and $E$, as before, but replacing $C$ used so far by $C_3$, we find that $\mathsf{J}_D$ and $\mathsf{J}_E$ coincide on $C_3$ (as its dimension is $>2$). In particular, $\mathsf{J}_D(e)=\mathsf{J}_E(e)$. 
\end{itemize}

To finish the proof, we first note that Theorem \ref{Hamhalter1} holds for $A=B=\C$ by inspection, whereas the cases $A\cong B\cong  \C^2$ or $\cong M_2(\C)$ have already been discussed.

 In all other cases we define $\mathsf{J}:A_{\mathrm{sa}}\raw B_{\mathrm{sa}}$ by putting $\mathsf{J}(a)=\mathsf{J}_D(a)$ for any maximal abelian unital $C^*$-subalgebra $D$ containing $C=C^*(a)$ and hence $a$; as we just saw, this is independent of the choice of $D$. Since each $\mathsf{J}_D$ is an isomorphism of commutative $C^*$-algebras, $\mathsf{J}$ is a weak Jordan isomorphism. Finally, uniqueness of $\mathsf{J}$ (under the stated restriction on $A$) 
follows from Lemma \ref{Hamlemma}.\enp
\end{proof}

Theorem \ref{Hamhalter1} begs the question if we can strengthen weak Jordan isomorphisms to Jordan isomorphism. This hinges on the extendibility of weak Jordan isomorphisms to linear maps (which are automatically Jordan isomorphisms). {The problem of whether a quasi-linear map (i.e., a map that is linear on commutative subalgebras) is linear has been studied by Bunce and Wright (1996) for the case of C*-algebras without a quotient isomorphic to $M_2(\mathbb C)$. More suitable for our setting is a generalization of Gleason's Theorem, proven by Bunce and Wright (1992), and thorougly discussed in (Hamhalter, 2004). More generally, one can rely on Dye's Theorem for AW*-algebras. For the exact statement and its proof we refer to Hamhalter (2015), who used it to prove the following result:} 
\begin{theorem}\L{Hamhalter2}
Let $A$ and $B$ be unital $AW^*$-algebras, where $A$ contains neither $\C^2$ nor $M_2(\C)$ as a summand. 
Then there is a bijective correspondence between order isomorphisms $\mathsf{B}: \CC(A)\raw\CC(B)$ and
 Jordan isomorphisms $\mathsf{J}:A_{\mathrm{sa}}\raw B_{\mathrm{sa}}$.
\end{theorem}
{We note that a version of this theorem in the setting of von Neumann algebras is proven in D\"oring and Harding (2010).}
If $A=B=B(H)$, then the ordinary Gleason Theorem suffices to yield the crucial lemma for 
Wigner's Theorem for Bohr symmetries (i.e.\ Theorem \ref{BigTheorem}.6).
\section{Projections}\L{secPA}
Given some $C^*$-algebra $A$, the orthocomplemented poset $\CP(A)$ of projections in $A$ satisfies the following two conditions:
\begin{enumerate}
	\item if $p\leq q^\perp$, then $p\vee q$ exists (and is equal to $p+q$ in $A$);
	\item if $p\leq q$, then $q=p\vee(p^\perp\wedge q).$
\end{enumerate}
We say that $\CP(A)$ is an \emph{orthomodular} poset. We note that $\CP(A)$ is Boolean if $A$ is commutative, but the converse implication does not hold. Indeed, Blackadar (1981) showed the existence of a non-commutative $C^*$-algebra whose only projections are trivial (and hence form a Boolean algebra). 
Notwithstanding such cases, it might be interesting to investigate in which ways $\CC(A)$ and $\CP(A)$ determine each other. In one direction we have the following result (Lindenhovius, 2016):
\begin{theorem}\label{thm:CAdeterminesProjA}
	Let $A$ and $B$ be $C^*$-algebras. Then any order isomorphism $\CC(A)\to\CC(B)$ induces an isomorphism of orthomodular posets $\CP(A)\to\CP(B)$.
\end{theorem}

Its proof is based on the following observations. Firstly, a $C^*$-algebra $A$ is called \emph{approximately finite-dimensional}, abbreviated by \emph{AF}, if there is a directed collection $\DD$ of finite-dimensional $C^*$-subalgebras of $A$ such that $A=\overline{\bigcup\DD}$. A commutative $C^*$-algebra $A$ is AF if and only if it is generated by its projections, which form a Boolean algebra, since $A$ is commutative. One can show that the Gelfand spectrum of $A$ is homeomorphic to the Stone spectrum of $\CP(A)$, and that there is an equivalence between the category of Boolean algebras with Boolean morphisms and the category of commutative AF-algebras with *-homomorphisms.

Given a $C^*$-algebra $A$, we let $\CC_\AF(A)$ be the subposet of $\CC(A)$ consisting of all commutative $C^*$-subalgebras of $A$ that are AF. There are several order-theoretic criteria whether or not an element $C\in\CC(A)$ belongs to $\CC_\AF(A)$, which is important since any order-theoretic criterion assures that an order isomorphism $\CC(A)\to\CC(B)$ restricts to an order isomorphism $\CC_\AF(A)\to\CC_\AF(B)$. Firstly, $C\in\CC_\AF(A)$ if and only if $C$ is the supremum of a directed subset of the \emph{compact} elements of $\CC(A)$, where $B\in\CC(A)$ is called compact if for each directed subset $\DD$ of $\CC(A)$ the inclusion $B\subseteq\bigvee\DD$ implies that $B\subseteq D$ for some $D\in\DD$;  note that the supremum $\bigvee\DD$ of any directed subset $\DD$ exists, and is given by $\overline{\bigcup\DD}$.
The compact elements of $\CC(A)$ turn out to be the finite-dimensional elements of $\CC(A)$, so clearly $C$ is AF if and only if it is the directed supremum of compact elements. Another criterion for $C\in\CC_\AF(A)$ is that $C$ is the supremum of some collection of atoms in $\CC(A)$. The intuition behind this criterion is that $C\in\CC_\AF(A)$ if and only if it is generated by its projections, and any atom of $\CC(A)$ is a $C^*$-subalgebra of $A$ that is generated by single proper projection in $A$. For details, we refer to Heunen \& Lindenhovius (2015).

Secondly, given some orthomodular poset $P$, we say that some subset $D\subseteq P$ is a \emph{Boolean subalgebra} of $P$ if it is a Boolean algebra in its relative order for which the meet, join, and orthocomplementation agree with the meet, join, and orthocomplementation, respectively on $P$. We denote the poset of Boolean subalgebras of $P$ ordered by inclusion by $\CB(P)$. The equivalence between the categories of Boolean algebas and of commutative AF-algebras yields the following proposition:
\begin{proposition}\label{prop:equivCAFandBProj}\
	Let $A$ be a $C^*$-algebra. The map 
\begin{eqnarray}
\CC_\AF(A)&\to&\CB(\CP(A));\\
C&\mapsto &\CP(C),
\end{eqnarray}
  is an order isomorphism with inverse $B\mapsto C^*(B),$ where $C^*(B)$ denotes the $C^*$-subalgebra of $A$ generated by $B$.
\end{proposition}

A modification of the Harding--Navara Theorem (cf. Remark 4.4 in Harding \& Navara, 2011)
states that if $P$ and $Q$ are orthomodular posets, then an order isomorphism $\Phi:\CB(P)\to\CB(Q)$ is induced by an orthomodular isomorphism $\varphi:P\to Q$ via $B\mapsto\varphi[B]$. Moreover, this orthomodular isomorphism is unique if $P$ does not have \emph{blocks}, i.e., maximal Boolean subalgebras of four elements. In combination with the previous proposition, this proves the following  (Lindenhovius, 2016):
\begin{theorem}\label{thm:orderisoCAandorthoisoBProjA}
	Let $A$ and $B$ be $C^*$-algebras. Then for each order isomorphism $$\Phi:\CC_\AF(A)\to\CC_\AF(B)$$ there is an orthomodular isomorphism $\varphi:\CP(A)\to\CP(B)$ such that $$\Phi(C)=C^*(\varphi[\CP(C)]),$$ for each $C\in\CC_\AF(A)$, which is unique if $\CP(A)$ does not have blocks of four elements.
\end{theorem}
Theorem \ref{thm:CAdeterminesProjA} now is an easy consequence of Theorem \ref{thm:orderisoCAandorthoisoBProjA}.

\section*{Acknowledgement}
The first author has been supported by  Radboud University and Trinity College (Cambridge).
The second author was supported by the Netherlands Organisation for Scientific
Research (NWO) under TOP-GO grant no. 613.001.013.
%\newpage

\end{document}

%% file: BohrificationpreambleLL2.tex
\newcommand{\half}{\mbox{\footnotesize $\frac{1}{2}$}}
\newcommand{\hi}[1]{\emph{\textbf{#1}}}
\newcommand{\qm}{quantum mechanics}
\newcommand{\er}{\eqref}
\renewcommand{\L}{\label} 
\newcommand{\beq}{\begin{equation}}
\newcommand{\eeq}{\end{equation}} 
\newcommand{\bea}{\begin{eqnarray}}
\newcommand{\eea}{\end{eqnarray}}

\newcommand{\wed}{\wedge}
 \newcommand{\til}{\tilde}
\newcommand{\raw}{\rightarrow}

\newcommand{\lraw}{\leftrightarrow}

\newcommand{\x}{\times}

\newcommand{\ca}{C*-algebra}

\newcommand{\Hs}{Hilbert space}

   \newcommand{\vna}{von
Neumann algebra}

\newcommand{\CA}{\mathcal{C}(A)} 
\newcommand{\al}{\alpha} 
\newcommand{\gm}{\gamma}

\newcommand{\lm}{\lambda} 
\newcommand{\rh}{\rho} 
 \newcommand{\ta}{\tau} 
 \newcommand{\phv}{\varphi}
\newcommand{\ch}{\ch}  
\newcommand{\om}{\omega} 
\newcommand{\ups}{\upsilon}
\newcommand{\daw}{\!\downarrow}

\newcommand{\inv}{^{-1}}

\newcommand{\Tr}{\mbox{\rm Tr}\,}

\newcommand{\CB}{{\mathcal B}}
\newcommand{\CC}{{\mathcal C}} 
\newcommand{\DD}{{\mathcal D}} 
\newcommand{\CE}{{\mathcal E}}

\newcommand{\CO}{{\mathcal O}} \newcommand{\CP}{{\mathcal P}}

\renewcommand{\L}{\label}

\newcommand{\C}{{\mathbb C}} 
 \newcommand{\R}{{\mathbb R}}
\newcommand{\T}{{\mathbb T}} 
  \makeatletter
 \newcommand{\enp}{\hfill\mbox{}\qed}
\makeatletter
\def\moverlay{\mathpalette\mov@rlay}
\def\mov@rlay#1#2{\leavevmode\vtop{%
   \baselineskip\z@skip \lineskiplimit-\maxdimen
   \ialign{\hfil$\m@th#1##$\hfil\cr#2\crcr}}}
\newcommand{\charfusion}[3][\mathord]{
    #1{\ifx#1\mathop\vphantom{#2}\fi
        \mathpalette\mov@rlay{#2\cr#3}
      }
    \ifx#1\mathop\expandafter\displaylimits\fi}
\makeatother

\newcommand{\AF}{\ensuremath{\mathrm{AF}}}

\newtheorem{theorem}{Theorem}[section]
\newtheorem{definition}[theorem]{Definition}
\newtheorem{lemma}[theorem]{Lemma}
\newtheorem{proposition}[theorem]{Proposition}